\documentclass[11pt]{article}     
\usepackage[utf8]{inputenc}
\usepackage{float}
\usepackage{graphicx}
\usepackage{amsthm}
\usepackage{latexsym}
\usepackage{amsmath}
\usepackage{amssymb}
\usepackage{amsfonts}
\usepackage{multirow}
\usepackage{amsthm}
\usepackage{enumerate}
\usepackage{graphics}
\usepackage{hhline,xstring}
\usepackage{fullpage}
\usepackage{color}
\usepackage{comment}
\usepackage{bussproofs}
\usepackage{booktabs,tabularx}
\usepackage{xspace}
\usepackage{xcolor,colortbl}
\usepackage{mathrsfs}

\usepackage{tikz}
\usetikzlibrary{arrows.meta,calc,decorations.text,fadings,through,positioning,mindmap,shadows,shapes,positioning,patterns,arrows,decorations.pathmorphing}
\usetikzlibrary{shapes.multipart}

\usepackage{caption}
\newtheorem{Observation}{Observation}
\newtheorem{Proposition}{Proposition}
\newtheorem{Example}{Example}

\theoremstyle{Observation}

\newcommand{\RN}[1]{%
  \textup{\uppercase\expandafter{\romannumeral#1}}%
}

\newcommand{\myComment}[1]{}

\title{Comparing Knowledge: An Analysis of the Relative Epistemic Powers of Groups}
\author{Alexandru Baltag\footnote{Institute for Logic, Language and Computation, University of Amsterdam, A.Baltag@uva.nl.} $\,\,$ \& Sonja Smets\footnote{Institute for Logic, Language and Computation, University of Amsterdam, S.J.L.Smets@uva.nl.}}
\date{}
\begin{document}
\maketitle
\begin{abstract}
We use a novel type of epistemic logic, employing comparative knowledge assertions, to analyze the relative epistemic powers of individuals or groups of agents. Such comparative assertions can express that a group has the potential to (collectively) know everything that another group can know. Moreover, we 
look at comparisons involving various types of knowledge (fully introspective, positively introspective, etc.), satisfying the corresponding modal-epistemic conditions (e.g., $S5$, $S4$, $KT$).     
For each epistemic attitude, we are particularly interested in what agents or groups can know about their own epistemic position relative to that of others.
\smallskip\par\noindent
{\bf Keywords:} epistemic logic; multi-agent epistemic logic; epistemic comparison; logics for multi-agent systems; collective epistemic attitudes; comparative knowledge; group knowledge.
\end{abstract}

\section{Introduction} 

{We} use a novel type of epistemic logic to study scenarios in which agents gain access to other agents' information or data, and hence they can in principle learn everything known to those others. These are scenarios in which one individual or group of agents knows at least as much as another individual or group of agents. To~study these scenarios, we introduce specific epistemic comparison statements in our logical language. We are specifically interested in what these individual agents or groups of agents can deduce about their own epistemic position relative to that of others, who may know more or less than them. The~answer to this question is the topic of this paper and, as we will show, it depends on the type of epistemic attitude that we attribute to the individual agents and to groups of~agents.

\par\noindent

In this paper, we make use of different notions of knowledge, including both weaker and stronger notions, as~studied in the
context of standard epistemic {logic} 
 \cite{FHMV}. The~weakest attitude of knowledge we consider in this paper, assumes only the veracity of knowledge for agents who know all tautologies and the logical consequences of their knowledge. It is the concept studied in the modal logic known as $KT$. A stronger epistemic notion extends the veracity-requirement with ``positive introspection'' in the modal logic $S4$, i.e.,~knowledge is taken to be truthful and if these agents know a proposition then they know that they know it. The~strongest notion adds the requirement of ``negative introspection'' to $S4$, so if these agents do not know a proposition, they also know that they do not know it. The~strongest notion is studied in the modal logic $S5$ and is obtained by quantifying over all the possible alternatives that are consistent with the agent's available information. In~this sense, the~strongest form of knowledge in this paper is closely related to the epistemological conception of ``information-as-range'' (see~\cite{vB}).

\par\noindent

In addition to the notions of individual knowledge, we consider a number of collective attitudes. One of the central notions in epistemic logic is ``distributed knowledge'', which ``corresponds to knowledge that is distributed among the members of the group, without~any individual agent necessarily having it.''~\cite{HalpernMoses1984}. Distributed knowledge is viewed as a potential collective attitude, capturing the implicit knowledge in a group of agents. It is potential in the sense that it refers to the information that a group can come to know if they pull together all the individual knowledge of each member of the group and close it under logical consequence. Hence, the distributed knowledge of a proposition $\varphi$ in a group of agents $G$ can hold without any member in the group $G$ actually knowing $\varphi$, but~potentially each individual can come to know $\varphi$ if all the group members share their information. For~example, we can think about a group of experts coming from different fields, all of whom have some partial information as to how they can solve a problem. The~combined knowledge of these experts would be the key to the solution when we close it under logical deduction, which expresses the group's distributed knowledge. Other examples come from computer science where we reason about distributed systems. For~instance, in a multi-body planning task, the~information coming from different sensors can be collected together and acted upon in order to reach a specific~goal.

\par\noindent

Another collective attitude of central importance is ``common knowledge''. Common knowledge is described in~\cite{HalpernMoses1984} as a form of ``public knowledge''. Common knowledge is here viewed as the strongest epistemic attitude that a group can have; it requires at the start that the group has mutual knowledge of $\varphi$, i.e.,~that every individual in the group knows $\varphi$, but~also all levels of higher-order knowledge including that everybody in the group knows that everybody knows $\varphi$ and that everybody knows that, etc. Common knowledge, in~this hierarchical account---for the history of the development of this concept, see, e.g.,~\cite{SEPOnCommonKnowledge}---extends the notion of mutual knowledge with the requirement of ad infinitum levels of higher-order knowledge. Common knowledge is an important concept in different fields, it is related to J. McCarthy's concept of that what any fool knows (see~\cite{16,FHMV}) as studied in the context of common-sense reasoning in AI, but~it also refers to what is established as the common ground in linguistic scenarios and as that what can be coordinated~upon.

\par\noindent

In our analysis of individual and group comparative knowledge, we will distinguish between the case in which an individual or a group of agents implicitly knows at least as much as another individual or group, the~case in which they implicitly know strictly more than others, the~case in which they implicitly know the same things (i.e., they are epistemically equivalent), and the case in which they have different uncertainties and hence are epistemically incomparable. For~these different cases, we consider what agents individually or collectively know about their own epistemic position, using the variety of epistemic notions mentioned above. Our analysis will make essential use of the tools of multi-agent epistemic~logic.

\par\noindent

This paper is organized as follows. In~Section~\ref{sec2} we introduce the background information on the logic of distributed and common knowledge. In~Section~\ref{sec3} we proceed with the analysis of comparative knowledge. In~Section~\ref{sec4} we end the paper with a conclusion and some pointers for further~work.

\section{Group Epistemic Logic: Distributed, Common, and Comparative Knowledge} \label{sec2}

Given a group of agents whose distributed knowledge can be computed, the~question is if everyone contributes equally to the knowledge of the group? In this section we show that this is not necessarily the case. Indeed, it may well be that all the knowledge of some agent (or group of agents) is already fully covered by another agent (or group of agents). Some agent (or subgroup of agents) may know at least as much as another agent in the group (or other subgroup). The~analysis of such scenarios requires us to compare the knowledge of agents. We will proceed in this paper by using the tools of modal logic (see~\cite{BRV}) and, more specifically, epistemic logic (see~\cite{FHMV}).

\par\noindent

\subsection{Syntax and~Semantics}

In this section, we introduce the logic obtained by adding a ``group epistemic comparison'' operator $A\preceq B$ to the standard multi-agent epistemic logic with distributed and common knowledge from~\cite{FHMV}. The~idea is to capture the epistemic comparison between both individuals and groups of agents. Earlier work in~\cite{Barteld} already includes a comparative modal operator $a \preceq b$ to express that locally, in~a given possible world, all sentences that are known by agent $a$ are also known by agent $b$. In~\cite{LPAR}, we generalized this idea to groups of agents, such that the epistemic comparison statements of the form $A \preceq B$ express that the group of agents $A$ knows at least as much as group $B$, and~we axiomatized the resulting logic on the class of $S5$-models. Here, we will study such group epistemic comparisons in the more general context of epistemic models of three types ($KT$, $S4$, and $S5$).

\smallskip\par\noindent
\textbf{{Vocabulary:} 
 agents, groups, and~atoms.} Throughout this paper, we assume a fixed (finite) set ${\cal A}=\{a, b, c, \ldots\}$ of \emph{agents}, and~a set $\Omega=\{p,q,r,\ldots\}$ of \emph{atomic sentences}.
  We use capital letters $A, B, C, G, \ldots$ to denote \emph{groups} of agents, i.e.,~non-empty subsets of ${\cal A}$.

\smallskip\par\noindent
{\bf Syntax.} The language of the logic {$LDC\preceq$} of Distributed knowledge, Common knowledge, and Epistemic Group Comparisons is recursively defined by the following BNF form: 
$$
\begin{array}{cccccccccccc}
\varphi ::= &  p &|& \neg \varphi &|& \varphi \wedge \varphi &|& K_A\varphi &|& K^A \varphi &|& A \preceq B
\end{array}
$$

\par\noindent
Here, atoms $p$ denote ontic ``facts'' (of a non-epistemic nature); $\neg$ and $\wedge$ are classical negation and conjunction; $K_A \varphi$ says that $\varphi$ is group (distributed) knowledge in the group of agents $A \subseteq {\cal A}$; $K^A \varphi$ says that $\varphi$ is common knowledge in the group  $A$; while $A\preceq B$ is an epistemic comparison statement, saying that group $A$'s distributed knowledge includes group $B$'s distributed knowledge. In~general, when we refer to a group's ``knowledge'', we will mean the group's distributed knowledge. So $A\preceq B$ can be read as short for ``group $A$ knows at least as much as group $B$'' (or group $A$ knows all $B$ knows). Note that for groups consisting of single agents, we will skip the brackets and write $a \preceq b$ instead of $\{a\} \preceq \{b\}$. 
\par\noindent
Note that individual knowledge $K_a\varphi$ is not a primitive notion in this language: it will be defined as the distributed knowledge of the singleton-agent group $\{a\}$.

\smallskip\par\noindent
\textbf{Abbreviations.} We have the usual Boolean notions of disjunction $\varphi\vee \psi$, implication $\varphi\to\psi$, and bi-conditional $\varphi\leftrightarrow\psi$, and in addition, we adopt the following abbreviations (where $A\not\preceq B$ is itself an abbreviation for $\neg (A\preceq B)$):

\begin{table}[H]
\begin{center}
\resizebox{11cm}{!}{%
\begin{tabular}{ccccc}
(Individual knowledge) && $K_a \varphi$ & $:=$ & $K_{\{a\}} \varphi$ \ \\
(Strict comparison) && ${A\prec B}$ & $:=$ & $A\preceq B\wedge B\not\preceq A$ \ \\
(Epistemic equivalence) && $A \equiv B$ & $:=$ & $(A \preceq B) \wedge (B \preceq A)$ \ \\
(Epistemic incomparability) && $A \perp B$ & $:=$& $(A\not\preceq B) \wedge (B\not\preceq A)$ \ \\
\end{tabular}
}
\end{center}
\end{table}
\vspace{-0.5cm}
\par
To provide a semantic interpretation of the constructs of the language, we consider \emph{three types of epistemic models}: $KT$-models (the most general ones), $S4$-models, and $S5$-models. They differ by increasingly strong introspection assumptions (on individual knowledge):

\smallskip\par\noindent
{\bf Epistemic ($KT$) models.} An \emph{epistemic model} (also called a \emph{$KT$-model}) is just a reflexive multi-agent Kripke model (see~\cite{BRV}); i.e.,~a relational structure $M = (S, R_a, ||{\bullet}||)_{a\in {\cal A}}$, where

\begin{itemize}
\item $S$ is, as usual, a set of possible worlds;
\item for each agent $a \in {\cal A}$, $R_a\subseteq S\times S$ is a \emph{reflexive} accessibility relation, denoting \emph{epistemic possibility}: $sR_a w$ means intuitively that world $w$ is consistent with agent $a$'s knowledge in world $s$ (hence, $w$ is a ``possible'' alternative for $s$ from $a$'s perspective);
    \item the valuation or \emph{truth-assignment} function, $||{\bullet}||: \Omega \to \mathcal{P}(S)$, maps every element of the given set $\Omega$ of atomic sentences into sets of possible worlds. Intuitively, for~every given atomic sentence $p \in \Omega$, the~valuation $||p||$ tells us in which worlds the factual content of $p$ is true.
\end{itemize}

\par\noindent
The fact that we require $R_a$ to be reflexive corresponds to the assumption of Veracity of Knowledge (i.e., \emph{knowledge implies truth}: $K_a\varphi\to \varphi$), also known as the axiom $T$ in Modal~Logic.

\smallskip\par\noindent
{\bf $S4$ models.} An epistemic model is an \emph{$S4$ model} if all accessibility relations are \emph{preorders} (i.e., not only reflexive, but~also \emph{transitive}). Requiring $R_a$ to be transitive corresponds to the assumption of Positive Introspection (i.e., \emph{knowledge implies knowledge of knowledge}: $K_a\varphi\to K_a K_a\varphi$), also known as the KK Principle in Epistemology, or~axiom $4$ in Modal~Logic.

\smallskip\par\noindent
{\bf $S5$ models.} Finally, an~epistemic model is an \emph{$S5$ model} if all accessibility relations are \emph{equivalence relations} (i.e., \emph{reflexive, symmetric, and transitive}). In~the context of the other two conditions, symmetry of $R_a$ is equivalent to another relational property, called Euclideaness, which corresponds to the assumption of Negative Introspection (i.e., \emph{ignorance implies knowledge of ignorance}: $\neg K_a\varphi\to K_a \neg K_a\varphi$), also known as the axiom $5$ in Modal~Logic.

Of course, these classes are included in each other: $S5$-models are $S4$-models, and~$S4$-models are $KT$-models.

\smallskip\par\noindent
\textbf{Joint possibility and common-knowledge possibility relations.}
To define notions of group knowledge, we need to extend our accessibility relations to groups of agents $A \subseteq {\cal A}$. The~\emph{joint possibility} relation $R_A$ is defined as the intersection of the individual accessibility relations for the agents in the group $A$, i.e.,
$$R_A \,\, :=\,\, \bigcap_{a\in A} R_a$$

\par
Essentially, $s R_A w$ means that in world $s$, all agents in $A$ consider $w$ possible: none of them can exclude it based on their knowledge. Joint possibility $R_A$ is the accessibility relation underlying the standard notion of \emph{distributed knowledge} for group $A$: in this sense, $sR_A w$ means that world $w$ is consistent with group $A$'s distributed knowledge in world $s$. Dually, the~\textit{common-knowledge possibility relation} $R^A$ points towards the possibilities allowed by the group's common knowledge: $sR^A w$ means that world $w$ is consistent with group $A$'s common knowledge in world $s$. Formally, $R^A$ is the reflexive transitive closure
 \footnote{As our epistemic relations $R_a$ are already reflexive, this is the same as the transitive closure.} of the individual relations in group $A$, i.e.,

$$R^A \,\, :=\,\, (\bigcup_{a\in A} R_a)^\ast$$

\par
The relation $R^A$ will be the accessibility relation underlying the standard notion of \emph{common knowledge} for group $A$. Note that $R^A$ is always transitive, even when the underlying relations are not. Because~of this, while
 $R_{\{a\}}=R_a$ holds in all epistemic ($KT$) models, in~general we have $R^{\{a\}}\neq R_a$ (except when the models are $S4$). So individual knowledge will in general be a special case of distributed knowledge (of a singleton group $\{a\}$), but~\emph{not} a special case of common~knowledge.

\par\noindent
\begin{Observation}\label{sum}
The relation for the distributed knowledge of the union of two groups of agents $A \cup B$, i.e.,~$R_{A\cup B}$, is given by the intersection of the relations for each group, i.e.,~$R_A \cap R_B$.
\end{Observation}
\begin{proof}
Unfolding the definition of the relation for distributed knowledge we obtain $R_{A \cup B} =$ \linebreak ${\bigcap_{c \in A \cup B} R_c}$ which is given by $(\bigcap_{c \in A} R_c)\cap (\bigcap_{c\in B}R_c)$, i.e.,~$R_A \cap R_B$.
\end{proof}

\smallskip\par\noindent
{\bf Semantics.} Given any epistemic model $M = (S, R_a, ||{\bullet}||)_{a\in {\cal A}}$, we can define the \emph{satisfaction relation} $w\models_M \varphi$ (between worlds $w$ in $M$ and formulas $\varphi$); and when the model is fixed, we skip the subscript and simply write $w\models \varphi$.
The definition is by recursion on the complexity of formulas $\varphi$, via the following clauses:
$$w \models p \mbox{ iff } w \in ||p||$$
$$w \models \neg\varphi \mbox{ iff } w \not \models \varphi $$
$$w \models \varphi\wedge \psi \mbox{ iff } w  \models \varphi \mbox { and } w \models \psi $$
$$w \models K_A\varphi \mbox{ iff } \forall s \in S (w R_A s \mbox{ implies } s \models \varphi)$$
$$w \models K^A\varphi \mbox{ iff } \forall s \in S (w R^A s \mbox{ implies } s \models \varphi)$$
$$w \models A\preceq B \,\, \mbox{ iff }
\, \forall s\in S \,\, (w R_A s \, \mbox{implies} \, w R_B s).$$

\par\noindent
In words: the semantics of atoms is given by the valuation; the semantics of Boolean connectives is classical; $K_A$ and $K^A$ are defined as standard Kripke modalities for the relations $R_A$ and $R^A$; finally, $A\preceq B$ holds iff group $A$ (distributedly) knows all that group $B$ (distributedly) knows (or equivalently: every world accessed by group $A$'s joint possibility relation is also accessed by group $B$'s joint possibility relation).

\smallskip\par\noindent
\textbf{Validity ($KT$, $S4$, or $S5$).} We can define a notion of validity for each of our three classes of epistemic models: we say that $\varphi$ is \emph{$KT$-valid}, and~write $\models_{KT}\varphi$, if~$\varphi$ is true in all possible worlds in all $KT$-models; similarly, we define \emph{$S4$-validity} $\models_{S4}\varphi$, and~\emph{$S5$-validity} $\models_{S5}\varphi$.

\par\noindent
Note that the strict version of the epistemic comparison operator expresses that (locally, in~the actual world) group $A$ is  ``more expert'' than (or ``epistemically superior'' to) group $B$. The~epistemic equivalence statement $B \equiv C$
says that  the epistemic positions of $B$ and $C$ are equally strong (in the actual world). Finally, $A\perp B$ means that the two groups are epistemically incomparable: none of them is superior, but~they are not equally strong either. This leads to the following observation:
 \begin{Observation}\label{constructs} The following formulas are $KT$-valid (hence, also $S4$- and $S5$-valid):
 $${(B \not \preceq C)} \leftrightarrow {(C \prec B \vee B \perp C)}$$
 $$ {(B \not \prec C)} \leftrightarrow {(C \preceq B \vee B \perp C)}$$
 $$ {(B \not \perp C)} \leftrightarrow {(C \preceq B \vee B \preceq C))}$$
\end{Observation}

\par
The following examples illustrate how we represent propositions about the epistemic attitudes of agents, including their individual knowledge, their distributed, and their common knowledge. We focus in particular on the distinction between fully introspective agents (Example 1) and agents that are positive but not negative introspective (Example 2).

 \begin{Example}
\rm{Consider three agents, Alice $a$, Bob $b$, and Charles $c$, who reason about an opaque box containing both a penny and a quarter.\footnote{This is a variation of the example in~\cite{Hughes}, which uses a penny--quarter box to illustrate the logical properties of a classical physical system.} We consider four atomic propositions to describe the state of the box, $\Omega = \{H_1,T_1,H_2,T_2\}$, where the first indexed propositions ($H_1$ and $T_1$) represent the possible states that describe the upper face of the penny, and the second indexed propositions ($H_2$ or $T_2$) represent the possible states that describe the upper face of the quarter. The~state of the box is determined by the upper face of each coin. In~our example, we use the  abbreviations $HH, HT, TH, TT$ respectively to denote ($H_1 \wedge H_2,H_1 \wedge T_2,T_1 \wedge H_2,T_1 \wedge T_2$). The~three agents $a,b,c$ are taken to be fully (i.e., both positive and negative) introspective. No agent knows the state of the box, which is invisible to them, but~they each have some specific partial information. It is common knowledge in the group $\{a,b,c\}$ that agent $a$ is told the state of the penny, agent $b$ is told the state of the quarter, and agent $c$ is told if the coins match or differ. To~represent the epistemic states of these agents (individually or for any (sub)group), we draw the epistemic $S5$ model in Figure~\ref{Fig1.}.}
\end{Example}

\vspace{-0.3cm}
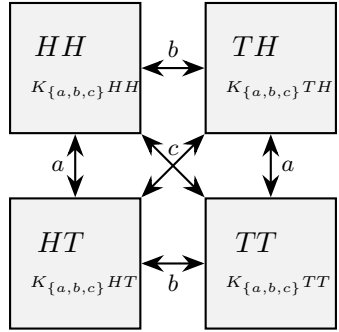
\begin{figure}[H]
\begin{center} 
\begin{tikzpicture}[node distance=2.6cm]
    \tikzstyle{w}=[draw=black, fill=gray!10, thick, text width=5em, text centered, minimum size=4.5em]
    \tikzstyle{m}=[draw=black, fill=gray!10, thick, text width=5
    em, text centered, minimum size=4.5em]
    \tikzstyle{every edge}=[draw, thick, font=\footnotesize]
    \tikzstyle{every label}=[font=\footnotesize]

    \node[m,text width=1.2cm] (w1) {$HH \quad$ {\tiny{$ {K_{\{a,b,c\}} HH}$}}};
    \node[w, right of=w1,text width=1.2cm] (w2) {$TH \quad$ {\tiny{$  {K_{\{a,b,c\}} TH}$}}};
    \node[w, below of=w1,text width=1.2cm] (w3) {$HT \quad$ {\tiny{$  {K_{\{a,b,c\}} HT}$}}};
    \node[w, below of=w2,text width=1.2cm] (w4) {$TT \quad$ {\tiny{$  {K_{\{a,b,c\}} TT}$}}};

    \path (w1) edge[{Stealth[length=3mm]}-{Stealth[length=3mm]}] node[above] {$b$} (w2);
    \path (w1) edge[{Stealth[length=3mm]}-{Stealth[length=3mm]}] node[left] {$a$} (w3);
    \path (w1) edge[{Stealth[length=3mm]}-{Stealth[length=3mm]}] node[above] {$c$} (w4);
    \path (w2) edge[{Stealth[length=3mm]}-{Stealth[length=3mm]}] node[above] {$\,$} (w3); 
    \path (w2) edge[{Stealth[length=3mm]}-{Stealth[length=3mm]}] node[right] {$a$} (w4);
    \path (w3) edge[{Stealth[length=3mm]}-{Stealth[length=3mm]}] node[below] {$b$} (w4);
  \end{tikzpicture}
  \end{center}
 \vspace{-0.2cm}
\caption{The drawing represents the epistemic $S5$ model for the penny--quarter box with three agents and a possible world for each of the four possible states of the box. All accessibility relations of agents are drawn except for the reflexive loops. At~each world, we specified (a conjunction of) the atomic facts that are true and we gave an example of a true epistemic proposition of the form $K_{\{a,b,c\}} \varphi$.}
\label{Fig1.}
\end{figure}

The accessibility relations in Figure~\ref{Fig1.} indicate that no matter which possible world represents the actual state of the box, nobody knows the true state. It is common knowledge in this model that each of the agents has some specific information about the state of the penny--quarter box. Indeed, it is common knowledge in the group $\{b,c\}$ that agent $a$ knows whether the first coin is heads up or tails up, i.e.,
$K^{\{b,c\}} (K_a (H_1 \wedge \neg T_1) \vee K_a (T_1 \wedge \neg H_1))$.
Note that this does not imply that agents $b$ or $c$ know which of these two disjuncts $a$ knows.
Moreover, in~the top-left world, $HH$ is distributed knowledge in any $2$-agent group: {\bf $K_{\{a,b\}} HH \wedge K_{\{b,c\}} HH \wedge K_{\{a,c\}} HH$} and in the $3$-agent group {\bf $K_{\{a,b,c\}} HH$}.  Observe in this example that if at least two agents pull all their information together, they will know the true state (heads or tails) of both coins. One verifies that indeed if $a$ and $b$ share what they know, they find out the true state of the box, and~similarly $a$ can find out what the true state of the box is by talking only to $c$. Hence, in~each possible world, any $2$-agent group knows the truth about the box, but no single agent~does.

\begin{Example}
 \rm{In a variation of the same scenario, we consider now two agents, $a$ and $b$, who are positive but not negative introspective, when reasoning about the state of the penny--quarter box. We draw an $S4$ model in Figure~\ref{Fig2.}, in~which we can verify that the agents can be ignorant without knowing it (i.e., negative introspection fails). For~instance, in~world $s$, agent $a$ does not know $TT$, i.e.,~$s \models \neg K_a TT$, and~she does not know her ignorance, i.e.,~$s \models \neg K_a \neg K_a TT$ (as she considers possible world $t$, in~which she \emph{would} know $TT$). Similarly, in~the same world $s$ agent $b$ does not know $HT$, i.e.,~$s \models \neg K_b HT$, but~he does not know his ignorance, i.e.,~$s \models \neg K_b \neg K_b HT$ (as he considers a possible world $v$, in~which he \emph{would} know $HT$). Note also that in world $s$ the group $\{a,b\}$ has distributed knowledge that $\neg(TT)$ is the case, i.e.,~$s \models K_{\{a,b\}}\neg (TT)$, but~they cannot exclude (even collectively, as~a group) any of the other three possibilities. Finally, note that in this model, it is common knowledge that $b$ knows whether $TT$ or $\neg(TT)$, i.e.,~we have $K^{\{a,b\}}(K_b TT \vee K_b \neg(TT))$.}
\end{Example}

\vspace{-0.3cm}
\begin{figure}[H]
\begin{center}
\begin{tikzpicture}[node distance=2.8cm]
    \tikzstyle{w}=[draw=black, fill=gray!10, thick,  minimum size=1.8em]
    \tikzstyle{m}=[draw=black, fill=gray!10, thick,  minimum size=1.8em]
    \tikzstyle{every edge}=[draw, thick, font=\footnotesize]
    \tikzstyle{every label}=[font=\footnotesize]

\tikzstyle{labelAgente}=[sloped,font = \scriptsize]

    \node[m,label = {[labelAgente]above:${\bf s}$},text width=1.2cm] (w1) {$\quad HH \quad \quad$ };
    \node[w,label = {[labelAgente]above:${\bf t}$},right of=w1,text width=1.2cm] (w2) {$\quad   TT \quad \quad$ };
    \node[w, label = {[labelAgente]below:${\bf u}$}, below of=w1,text width=1.2cm] (w3) {$\quad TH \quad \quad$ };
    \node[w, label = {[labelAgente]below:${\bf v}$}, below of=w2,text width=1.2cm] (w4) {$\quad  HT \quad \quad$ };

    \path (w1) edge[-{Stealth[length=3mm]}] node[above] {$a$} (w2);
    \path (w1) edge[-{Stealth[length=3mm]}] node[left] {$a,b$} (w3);
    \path (w1) edge[-{Stealth[length=3mm]}] node[above,outer sep=0.2cm,label distance=0.3cm] {$a,b$} (w4);
        \path (w3) edge[-{Stealth[length=3mm]}] node[above] {$b$} (w4);



  \end{tikzpicture}
   \end{center}
   \vspace{-0.5cm}
   \caption{The drawing represents the epistemic $S4$ model with four possible worlds, $s,t,u,v$, and two agents, $a,b$, for the penny--quarter box. Except~for the reflexive loops, all other relations are~drawn.}
    \label{Fig2.}
\end{figure}
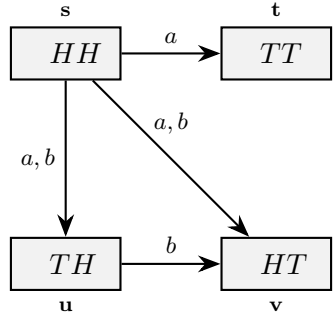
\par\noindent

We will use the provided models in these examples to reason about epistemic comparison statements in the following~section.

\subsection{Axiomatisation of $LDC\preceq$}

In~\cite{LPAR}, we provided a \emph{complete axiomatization over $S5$ models} for the logic {$LDC\preceq$} with distributed knowledge, common knowledge, and epistemic comparison.\footnote{The $LD\preceq$ fragment of our logic over $S5$ models is related to the logic LFD of functional dependence in~\cite{BvB2020a}, we refer to~\cite{LPAR} for details.} Some of these are still valid over epistemic ($KT$) models, so we will refer to them as ``axioms'', even if we do \emph{not} prove here any completeness result. In~Table~\ref{tb0}, we provide a full overview of the axioms and rules for $LDC\preceq$.

\begin{table}[H]
\begin{center}
\resizebox{15cm}{!}{%
\begin{tabularx}{\textwidth}{ll}
\toprule
\textbf{{(I)} 
} & \textbf{Axioms and rules of classical propositional logic}
 \vspace{1mm} \ \\
 \textbf{(II)} &\textbf{$KT$ axioms and rules for distributed knowledge}:   \ \\
($K_A$-Necessitation)& From $\varphi$, infer $K_A \varphi$  \ \\
($K_A$-Distribution) & $K_A (\varphi \to \psi) \to (K_A\varphi \to K_A\psi)$  \ \\
(Veracity) & $K_A\varphi \to \varphi$   \ \\
\textbf{(III)} & \textbf{Axioms for comparative knowledge}:  \ \\
(Inclusion) & $A\preceq B$, provided that $B \subseteq A$ \ \\
(Additivity) & $\left(A\preceq B\wedge A\preceq C\right) \to A\preceq B\cup C$  \ \\
(Transitivity) & $\left(A\preceq B\wedge B\preceq C\right)\to A\preceq C$ \ \\
(Knowledge Transfer) &  $A\preceq B \to \left(K_B\varphi \to K_A\varphi \right)$
\ \\
\textbf{(IV)} & \textbf{Axioms and rules for common knowledge}: \ \\
($K^A$-Necessitation) & From $\varphi$, infer $K^A \varphi$  \ \\
($K^A$-Distribution) & $K^A(\varphi\to \psi)\to (K^A\varphi\to K^A\psi)$  \ \\
($K^A$-Fixed Point) & $K^A\varphi \to (\varphi\wedge \bigwedge_{a\in A} K_a K^A\varphi)$  \ \\
($K^A$-Induction) & $K^A (\varphi \to \bigwedge_{a\in A} K_a \varphi)\to (\varphi\to K^A\varphi)$ \ \\
\textbf{(V)} & \textbf{Special axioms for $S4$-models}: \ \\
(Positive Introspection) & $K_A \varphi \to K_A K_A \varphi$   \ \\
\textbf{(VI)} & \textbf{Special axioms for $S5$-models}: \ \\
(Negative Introspection) & $\neg K_A\varphi \to K_A \neg K_A \varphi$   \ \\
(Known Superiority) & $A\preceq B \to K_A (A\preceq B)$  \ \\
\bottomrule
\end{tabularx}
}
\end{center}
\vspace{-0.2cm}
\caption{The table represents the Axioms (I--IV) for $\mathbf{LDC\preceq}$ over $KT$-models. For~$S4$-models, we need to add group (V) (Positive Introspection). For~$S5$-models, we need to add both groups $(V)$ and $(VI)$. Here, individual knowledge is an abbreviation $K_a \varphi := K_{\{a\}}\varphi$. Only the full system for $S5$-models is known to be~complete.}\label{tb0}
\end{table}

The $KT$, $S4$, and $S5$ axioms in Table~\ref{tb0} are well known. The~Inclusion axiom expresses that a larger group knows at least as much as
its subgroups (and so, in particular, the~distributed knowledge of a group includes the individual knowledge of each of its members). Additivity says that, if~a group knows at least as much as two other groups, it will know at least as much as their union. Transitivity simply says that epistemic comparison is transitive.  The~Knowledge Transfer axiom captures the intuitive meaning of $A \preceq B$:  when $A$ knows at least as much as $B$, then $B$ knows everything that $A$ knows.
The ``Known Superiority'' axiom states that epistemically ``superior'' groups know their superiority status: a group $A$ (relative to group $B$) always knows whether it knows at least as much as $B$ (or not).

\begin{Proposition}\label{p-1} All axioms in groups (I)--(IV) are sound on $KT$ models (i.e., $KT$-valid). Positive Introspection is sound on $S4$ models (thus, also on $S5$ models), while Negative Introspection and Known Superiority are valid only on $S5$-models.
\end{Proposition}

The fact that ``known Superiority'' is valid in $S5$ but fails in $S4$ will be shown in the next section. Furthermore, note that there is \emph{no ``Known Inferiority'' analogue}: $A\preceq B$ does \emph{not} imply $K_B (A\preceq B)$. We will provide the relevant counterexamples~below.

\par\noindent

In addition to these axioms, we consider in this paper also the following $KT$-validity:
$$\mbox{(group monotonicity)} \, \,\,\,\,\,  \models_{KT} K_B \varphi \to K_{B\cup C} \varphi$$
\par\noindent
Group monotonicity is usually taken as an axiom in systems for distributed knowledge, but~in our setting it is \emph{provable from the above axioms} (on $KT$ models).

Before we analyze the validities involving statements about knowledge comparison in the next section, we first use the above examples to illustrate how one can evaluate different epistemic comparison statements in these~models.

\smallskip\par\noindent
{\bf{Analysis of Example 1.} 
} We analyze a number of epistemic comparison statements in the $S5$-model of Example 1. Locally, whatever is the true state of the coins, every 2-agent group knows more than any other single agent in the group. For~example, together, $a$ and $b$ as a group, i.e.,~${\{a,b\}}$, knows more about the coins than ${c}$. Observe indeed that $\{a,b\} \preceq c$ but $c \not\preceq \{a,b\}$ holds in all possible worlds. Moreover, this epistemic statement is common knowledge among all agents in all four worlds. Observe in Figure~\ref{Fig1.} that $\{a,b\} \prec c$ is true in all worlds and that it can be reached by any concatenation of accessibility relations for any of the agents in the group $\{a,b,c\}$, hence in all worlds we have $K^{\{a,b,c\}} (\{a,b\} \prec c)$. If~$b$ would talk to both $a$ and $c$, then the two groups $\{a,b\}$ and $\{b,c\}$ would know the same information about the coins. Hence these groups have the same distributed knowledge, so they are ``epistemically equivalent'' $\{a,b\} \equiv \{b,c\}$ in all worlds, i.e.,~$\{a,b\} \preceq \{b,c\}$  and $\{b,c\} \preceq \{a,b\}$.

Epistemic comparison statements are not always commonly known in epistemic $S5$ models. We illustrate this with a slight variation of the above example:

\begin{Example}
\rm{We continue the above story of Example 1, with~the three fully introspective agents who reason about the state of a box containing a penny and a quarter. In~this example, the three agents are publicly told in advance that $TT$ is not possible. Hence the agents reason about the three remaining options: $HH,HT,TH$. As~before, it is common knowledge that agent $a$ is told the state of the penny, agent $b$ knows the state of the quarter, and $c$ knows if the coins match or differ. We draw the epistemic $S5$ model in Figure~\ref{Fig5.}.}
\end{Example}
\vspace{-0.2cm}
\begin{figure}[H]
\begin{center}

\begin{tikzpicture}[node distance=2.8cm]
        \tikzstyle{zz}=[decorate,decoration={zigzag,post=lineto,post length=5pt}]
    \tikzstyle{w}=[draw=black, fill=gray!10, thick, minimum size=1.6em]
    \tikzstyle{m}=[draw=black, fill=gray!10, thick, minimum size=1.6em]
        \tikzstyle{every edge}=[draw,thick,font=\footnotesize]

        \tikzstyle{every label}=[font=\footnotesize]

        \tikzstyle{ev}=[anchor=center,node distance=3.8cm]

        \tikzstyle{wred}=[w,draw=black]


\tikzstyle{labelAgente}=[sloped,font = \scriptsize]

        \node[w,label = {[labelAgente]above:${\bf s}$}, text width=1.6cm] (w1) {{\centerline{$HT$}} {\centerline{{\tiny{${\{b\}\prec \{a\}}$}}}}};

        \node[w,label = {[labelAgente]above:${\bf t}$},text width=1.6cm, right of=w1] (w2) {{\centerline{$HH$}} {\centerline{{\tiny{${\{a\}\not \preceq \{b\}}$}}}}};

        \node[w,label = {[labelAgente]above:${\bf u}$}, text width=1.6cm, right of=w2] (w3) {{\centerline{$TH$}} {\centerline{{\tiny{$ {\{a\}\prec \{b\}}$}}}}};


        \path (w1) edge[{Stealth[length=3mm]}-{Stealth[length=3mm]}] node[below,outer sep=-2pt]{${{a}}$} (w2);

        \path (w2) edge[{Stealth[length=3mm]}-{Stealth[length=3mm]}] node[below,outer sep=-2pt]{${{b}}$} (w3);

             \path (w1) edge[{Stealth[length=3mm]}-{Stealth[length=3mm]},bend right=20] node[below,outer sep=-2pt,]{${{c}}$} (w3);

      \end{tikzpicture}
      \end{center}
      \vspace{-0.2cm}
      \caption{The drawing represents the epistemic $S5$ model $M$ with three possible worlds, $s,t,u$, for the penny--quarter box in which agents are told that $TT$ is not possible. Except~for the reflexive loops, all other accessibility relations are~drawn.}
\label{Fig5.}
\end{figure}
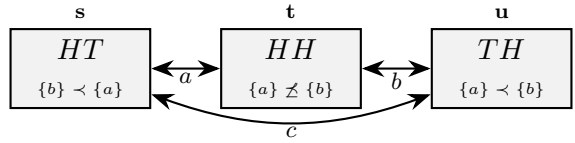

In the $S5$ model $M$ in Figure~\ref{Fig5.}, if~the real state of the box is $u$, then agent $a$ knows it (she knows the state of the first coin and she knows that $TT$ is not possible), but~$b$ does not know it and neither does $c$. In~particular, if~the state of the box is $u$, $a$ knows more than $b$, i.e.,~$u \models_{M} {a} \prec {b}$. In~the possible world $s$, $HT$ is true, here the epistemic powers are reversed as $b$ is epistemically leading over $a$, i.e.,~$s \models_{M} {b} \prec {a}$. If $HH$ is true, the~agents $a$ and $b$'s epistemic states are incomparable, i.e.,~$t\models_{M} {a} \perp {b}$.

\smallskip\par\noindent
{\bf{Analysis of Example 2.}} In the model of Figure~\ref{Fig2.}, observe that in world $s$, agent $b$ knows $\neg(TT)$ and, as such, $b$ knows more than agent $a$ at $s$, i.e.,~$s \models {b}\prec {a}$. In~world $u$, agent $a$ knows more than $b$, i.e.,~$u \models {a} \prec {b}$, while in the other worlds the agents are epistemically~equivalent.

\section{Reasoning About Epistemic~Comparison} \label{sec3}

In this section we will use the introduced models and in particular Proposition 1 to analyze the validity of a variety of epistemic statements about agents' epistemic comparative~position.

\subsection{Knowledge, Distributed Knowledge, and Common Knowledge About Epistemic Comparison in $S5$ Models}

First, let us check the $S5$-vality of the ``known superiority axiom'' in~its weak~form.

\begin{Proposition}\label{proposition2} In  $S5$ models, {\it if a group $B$ (collectively) knows at least as much as another group $C$, then $B$ (collectively) knows this fact}:
$$\models_{S5} B\preceq C \,\,\to \,\, K_B (B\preceq C)$$
\end{Proposition}
\begin{proof}
To show this, let us assume that we have $s \models B \preceq C$, at~some world $s$ of an $S5$ model $M$. We have to show that $s \models K_B (B\preceq C)$. For~this, let $w$ be any world such that $s R_B w$, and~we have to show that $w\models B\preceq C$. To~prove this, let now $t$ be any world with $w R_B t$, and~we have to show that $w R_C t$. To~prove this, we first note that $s R_B w$ and $w R_B t$ give us $s R_B t$ (by Transitivity). Second, from~$s R_B w$ and $s \models B \preceq C$ we derive that $s R_B w$ (by the semantics of $\preceq$), and~similarly from $s R_B t$ and $s \models B \preceq C$ we derive that $s R_B t$. Finally, $s R_C w$ and $s R_C t$ imply that $w R_C t$ (by the Euclideanness of $R_C$).   
\end{proof}

We continue by considering the scenario in which two groups are epistemically equally strong. Take for instance two groups, $B$ and $C$, who are epistemically equally strong in an epistemic $S5$ model, then each group knows it.
\begin{Proposition}\label{proposition3} In  $S5$ models, if~two groups $B$ and $C$ are epistemically (not, respectively) equivalent, then each group knows that this is (not, respectively) the case:
 $$\models_{S5} B \equiv C \, \to \, K_{B}(B \equiv C) \wedge K_{C}(B \equiv C) $$
 $$\models_{S5} B \not \equiv C \to K_B (B \not \equiv C) \wedge K_C (B \not \equiv C)$$ 
\end{Proposition}
\begin{proof}\label{proof2}
{{We start with the first statement. First, observe that ${\models_{S5} (B \equiv C)} \leftrightarrow {(B \preceq C \wedge C \preceq B)}$. We use the $S5$-validity of the known superiority axiom in Proposition \ref{proposition2} for both parts of the conjunction, i.e.,~${\models_{S5} (B\preceq C)} \to {K_B(B \preceq C)}$ and ${\models_{S5} (C \preceq B)} \to {K_C (C \preceq B)}$. Hence, ${\models_{S5} (B \equiv C)} \to (K_B (B \preceq C) \wedge K_C (C \preceq B))$. By~the Knowledge Transfer axiom and propositional logic, we then also have ${\models_{S5} (B \equiv C)} \to (K_C (B \preceq C)) \wedge (K_B (C \preceq B))$.
Putting both together, we have
${\models_{S5} B \equiv C} \to {K_B ((C \preceq B) \wedge (B \preceq C)) \wedge K_C ((C \preceq B) \wedge (B \preceq C))}$ and hence $\models_{S5} B \equiv C \, \to \, K_{B}(B \equiv C) \wedge K_{C}(B \equiv C)$.
\par\noindent
To prove the second statement, we use the theorem in $S5$ which states that whenever $\varphi \to K \varphi$ is provable then so is $\neg \varphi \to K \neg \varphi$, together with the proof of the first statement in this proposition.
}}
\end{proof}

We use Proposition \ref{proposition3} to show that the strict version of the statement in Proposition \ref{proposition2} is also $S5$-valid:
\begin{Proposition}\label{proposition4}
In  $S5$ models, {\it a more-expert group (collectively) knows its own
epistemic position over a less-expert group}: 
$$\models_{S5} B\prec C \,\, \to \,\, K_B (B \prec C)$$
\end{Proposition}
\begin{proof}
We use the valid equivalence $(B \prec C) \leftrightarrow (B \preceq C \wedge B \not \equiv C)$. Hence we have to show that $(B \preceq C \wedge B \not \equiv C) \to K_B (B \preceq C \wedge B \not \equiv C)$. Using Proposition \ref{proposition2}, we derive from the premise that $K_B (B\preceq C)$, i.e.,~the first conjunct in the conclusion. For~the second conjunct, we use the second statement in Proposition \ref{proposition3}, i.e.,
$\models_{S5} B \not \equiv C \to K_B (B \not \equiv C) \wedge K_C (B \not \equiv C)$. Together this yields $\models_{S5} B\prec C \to K_B (B \prec C)$.
\end{proof}

What if the groups are epistemically incomparable? It turns out that incomparability can be opaque to all the participants!
\begin{Observation}\label{observation3}
In $S5$ models, epistemic incomparability between two groups is not necessarily known by any of the two:
$$\not \models_{S5} B\perp C \,\,\to \,\, (K_B (B\perp C) \vee K_C (B \perp C))$$
\end{Observation}
\begin{proof}
Figure~\ref{Fig5.} provides an $S5$ model $M$ in which $t \models {a} \perp {b}$ while  $t \models \neg K_a ({a} \perp {b}) \wedge$ ${\neg K_b ({a} \perp {b})}$.
\end{proof}

We go on to show that the less-expert group does not necessarily know that another group is the more-expert group in epistemic $S5$ models. This sounds reasonable because several scenarios can be thought of where that is indeed the case. For~instance, in~the case of one's data being hacked, typically the victim does not know it but can find out that someone gained access to all of their information.  This can be expressed formally in Observation \ref{observation4}.
\begin{Observation}\label{observation4} In $S5$ models, a~less-expert group $C$ (collectively) does not necessarily know its own epistemic position with respect to a more-expert group $B$. This is also the case when $B$ knows at least as much as $C$.
$$ \not \models_{S5} B\prec C \,\, \to \,\,  K_C (B\prec C)$$
$$ \not \models_{S5} B\preceq C \,\, \to \,\,  K_C (B\preceq C)$$
\end{Observation}
\begin{proof}
To prove the first item, Figure~\ref{Fig5.} provides a model $M$ in which $u\models_{M} {a} \prec {b}$ and $u \models_{M}K_a({a}\prec {b})$ while $u \not \models_{M}K_b ({a} \prec {b})$. This shows that the first statement is not valid in $S5$ models.

\par\noindent
{To prove the second item, we use the valid equivalence} $(B \preceq C) \leftrightarrow {((B \prec C) \vee (B \equiv C))}$. In case $B \equiv C$ then Proposition \ref{proposition3} shows that $K_C(B \preceq C)$ does follow, hence the only option that remains is $B \prec C$. To~show that it does not imply $K_C (B \preceq C)$, we refer back to the proof of the first statement in this observation.
\end{proof}

However, a~less-expert team $C$ can know something about its own position relative to $B$, at~least in case they do {\it not} know at least as much as $B$, i.e.,
\begin{Proposition}\label{rp1} In $S5$ models, if~it is not the case that a group $C$ (collectively) knows at least as much as another group $B$, then $C$ (collectively) know this fact:
$$ \models_{S5} C \not \preceq B \,\, \to \,\,  K_C (C \not \preceq B)$$
\end{Proposition}
\begin{proof}
{{This follows from Proposition \ref{proposition2} stating that in $S5$ models, ``Known Superiority'' is valid ${\models_{S5} C \preceq B} \, \to \, {K_C(C \preceq B)}$, together with the theorem in $S5$ which states that whenever $\varphi \to K \varphi$ is provable then so is $\neg \varphi \to K \neg \varphi$.}}
\end{proof}

We now reflect on whether the group's epistemic comparison status is also known by the bigger group, consisting of both the more-expert and less-expert groups. Indeed, if~$B$ is the more-expert group over a group $C$, then, when $B$ collectively knows a proposition $\varphi$, so does the bigger group $B \cup C$ that $B$ is a part of. This leads to the following:
\begin{Proposition}\label{proposition7} In $S5$ models, if~a group $B$ (collectively) knows at least as much, more, or~the same (respectively) as another group $C$, then the bigger group consisting of both $B$ and $C$ knows this fact:
$$\models_{S5} B \preceq C \, \to \, K_{B \cup C}(B \preceq C)$$
$$\models_{S5} B \equiv C \, \to \, K_{B \cup C}(B \equiv C)$$
$$\models_{S5} B \prec C \, \to \, K_{B \cup C}(B \prec C)$$
\end{Proposition}
\begin{proof}
To show the first claim, we make use of Proposition \ref{proposition2} stating the validity of ``Known Superiority'' ${\models_{S5} B \preceq C \to K_B (B \preceq C)}$ and the ``Group Monotonicity for Distributed knowledge'', i.e., ${\models_{S5} K_B \varphi \to K_{B\cup C} \varphi}$. Putting both together and using propositional logic, we obtain ${\models_{S5} B\preceq C}$ $\to K_{B \cup C} (B \preceq C)$. The second and third validity follow as special cases.
\end{proof}

Moreover, if the two groups are not epistemically equivalent, this fact is also distributed knowledge in the bigger group consisting of both, i.e.:
\begin{Proposition}\label{q0}  In $S5$ models, if~it is not the case that a group $B$ is epistemically equivalent to another group $C$, then the bigger group composed of $B$ and $C$ knows this fact:
$$\models_{S5} B \not \equiv C \, \to \, K_{B \cup C}(B \not \equiv C)$$
\end{Proposition}
\begin{proof}
{{This claim follows from the validities in Proposition \ref{proposition7} together with the following theorem in $S5$ which states that whenever $\varphi \to K\varphi$ is provable then so is $\neg \varphi \to K \neg \varphi$.}}
\end{proof}

\smallskip\par\noindent
{\bf Common Knowledge about Epistemic Comparison}.\footnote{The propositions with epistemic comparison in this section that involve common knowledge are restricted to comparison  statements between singleton  groups of  agents. A generalization to epistemic comparison statements between arbitrary groups of agents can be obtained but requires the use of the notion of common distributed knowledge. We comment on this in the conclusion of this paper.} If it happens to be the case that a less-expert agent does know that another agent knows at least as much, then this fact is common knowledge among the group of both agents in $S4$ (and $S5$) models. 

\begin{Proposition}\label{p17} In $S4$-models, if~agent $c$ knows that another agent $b$ knows at least as much as her, then this fact is common knowledge between the two of them: 
$$\models_{S4} K_c({b} \preceq {c}) \to K^{\{b,c\}}({b} \preceq {c})$$
\end{Proposition}
\begin{proof}
By axiom $4$, we have $\models_{S4} K_c({b} \preceq {c}) \to {K_c K_c({b} \preceq {c})}$ (1). 
 On~the other hand, using axiom $T$ and the ``Knowledge Transfer'' axiom, we obtain ${\models_{S4} K_c({b} \preceq {c}) \to}{K_b K_c ({b} \preceq {c}))}$ {(2)}. Using {(1) and (2)} and propositional logic, we get that $\models_{S4} K_c({b} \preceq {c}) \to (K_c K_c({b} \preceq {c}) \wedge K_b K_c ({b} \preceq {c}))$, then we apply Necessitation to obtain $\models_{S4} K^{\{b,c\}} (K_c({b} \preceq {c}) \to (K_c K_c({b} \preceq {c}) \wedge K_b K_c ({b} \preceq {c})))$. Finally, applying the Induction axiom yields the desired conclusion.
\end{proof}

We continue by considering the scenario in which two agents $b$ and $c$ are epistemically equally strong. 
\begin{Proposition}\label{proposition9b} In $S5$ models, if~agents $b$ and $c$ are epistemically equivalent, then this fact is  common knowledge in the group consisting of $b$ and $c$:
 $$\models_{S5} {b} \equiv {c} \, \to \, K^{\{b,c\}}({b} \equiv {c})$$
\end{Proposition}
\begin{proof}\label{proof2}
Using Proposition \ref{proposition3}, applied to singleton groups ${b}$ and ${c}$, we have

${\models_{S5} {b} \equiv {c} \to} {K_b (({c} \preceq {b}) \wedge ({b} \preceq {c})) \wedge K_c (({c} \preceq {b}) \wedge ({b} \preceq {c}))}$.

By Necessitation we have $\models_{S5} K^{\{b,c\}} ({{b} \equiv {c}} \to K_b ({{c} \preceq {b}} \wedge {{b} \preceq {c}}) \wedge K_c ({{c} \preceq {b}} \wedge {{b} \preceq {c}}))$. 

Using the Induction axiom for common knowledge, we obtain:

 ${\models_{S5} {b} \equiv {c}} \to K^{\{b,c\}} {(({c} \preceq {b}) \wedge ({b} \preceq {c}))}$, i.e.,~${\models_{S5} {b} \equiv {c}} \, \to \, {K^{\{b,c\}}({b} \equiv {c})}$.
\end{proof}

In a similar fashion, but~now using the theorem in $S5$ which states that whenever $\varphi \to K \varphi$ is provable then so is $\neg \varphi \to K \neg \varphi$, one can show that:
\begin{Proposition}\label{proposition9c} In $S5$ models, if~it is not the case that agents $b$ and $c$ are epistemically equivalent, then this fact is common knowledge in the group consisting of $b$ and $c$:
$$ \models_{S5} {b} \not \equiv {c} \, \to \, K^{\{b,c\}}({b} \not \equiv {c})$$
\end{Proposition}

\subsection{Reasoning About Overlapping Groups and Epistemic~Free-Riders}

\smallskip\par\noindent
{\bf Overlapping Groups.} We consider the case in which two groups, $B \cup A$ and $C \cup A$, overlap only in $A$, i.e.,~they have a subgroup of agents $A$ in common. For~example, we can think of scenarios in which $A$ are the ``double-spy agents'' belonging to two competing teams. As~one may expect, the~double-agents $A$ may play a crucial role in establishing one team's epistemic position with respect to another team. More specifically, the~competitive epistemic advantage of a group $B \cup A$ over $C \cup A$ in the presence of the double agents $A$, does not imply that the group $B$ necessarily will have the same epistemic advantage over $C$ without $A$, as~it may well be the case that $B$ and $C$ are themselves epistemically equal or incomparable. The~latter case is expressed in the following observation:
\begin{Observation}\label{observation5} In $S5$ models, if~a group $B \cup A$ knows at least as much as another group
$C \cup A$ then it is not necessarily the case that group $B$ (collectively) also knows at least as much as group
$C$ without the presence of $A$:
$$\not \models_{S5} B \cup A \preceq C \cup A \, \to \,  B \preceq C$$
\end{Observation}
\begin{proof}
To show that the statement is not valid in $S5$ models, consider the model $M$ in Figure~\ref{Fig5.} in which agents $a$ and $b$ are epistemically incomparable in world $t$ but if they all team up with agent $c$, the~groups become epistemically equivalent, i.e.,~$t \models_{M} \{\{a\} \cup \{c\}\} \equiv \{\{b\} \cup \{c\}\}$ while $t \not \models_{M} a \preceq {b}$.
Hence, the~statement is also not valid in $S4$ or $KT$ models.
\end{proof}

However, if~a group $B$ knows at least as much as a different group $C$, then enlarging both groups with $A$ will not make an epistemic difference. This indicates that the converse statement of Observation \ref{observation5} is valid, moreover it is valid in weaker $KT$ models. Similarly, such a statement is valid for the case of epistemically equivalent groups:
\begin{Proposition}\label{KT-version} In $KT$-models, if~a group $B$ (collectively) knows at least as much (knows the same information, respectively) as~another group $C$, then $B \cup A$ knows at least as much (knows the same information, respectively) as~$C \cup A$:
$$\models_{KT} B \preceq C \, \to \,  B \cup A \preceq C \cup A$$
$$\models_{KT} B \equiv C \, \to \,  B \cup A \equiv C \cup A$$
\end{Proposition}
\begin{proof}
For the first statement in this proposition, first assume that $s \models B \preceq C$ in an arbitrary world $s$ of a $KT$ model, this means via the semantic clause of epistemic comparison statements in $KT$-models that $\forall t\in S \,\, (s R_B t \, \mbox{implies} \, sR_C t)$. We have to prove that ${\forall t\in S} \,\, (s R_{B \cup A} t \, \mbox{implies}$ $sR_{C \cup A} t)$. Hence, assume that $s R_{B \cup A} t$, which by Observation \ref{sum} gives us $s (R_B \cap R_A) t$, i.e.,~$s R_B t$ and $s R_A t$. Together they imply that $s (R_C \cap R_A) t$, i.e.,~$s R_{C \cup A} t$ and hence $s \models B \cup A \preceq C \cup A$. The~second statement of this proposition follows from this.
\end{proof}

From Proposition \ref{KT-version} also follows the slightly weaker statement $\models_{KT} B \prec C \, \to$  $B \cup A \preceq$ $C \cup A$.
Note that in this case when $B$ knows strictly more than $C$, it may well be that a group of agents $A$ can compensate for the epistemic difference between $B$ and $C$, which leads to:
\begin{Observation} In $S5$-models, a~more-expert group $B$ over another group $C$, may lose its relative epistemic position when both groups are extended with a group $A$. 
$$\not \models_{S5} B \prec C \, \to \,  B \cup A \prec C \cup A$$
\end{Observation}
\begin{proof}
To show that the statement in this observation is not valid in $S5$ models, we consider the model $M$ in Figure~\ref{Fig5.} where agent $a$ knows more than $c$ in world $u$, i.e.,~$u \models_{M} {a} \prec {c}$ but if each of them teams up with $b$, both teams will be epistemically equivalent, hence $s \models_{M} \{a,b\} \not \prec \{c,b\}$.
Hence, the~statement is also not valid in $S4$ or $KT$ models.
\end{proof}

If two different groups, $B$ and $C$, are epistemically comparable and one group has an epistemic advantage over the other whenever they each are fused with the agents in $A$, then the leading group will keep its epistemic advantage also without the presence of $A$. We can express this as the following proposition:
\begin{Proposition}\label{prop8} In $KT$-models, when $B$ and $C$ are epistemically comparable, the~epistemic position of a bigger more-expert group $B \cup A$ relative to group $C \cup A$, carries over to the smaller group $B$ relative to $C$ when $A$ has departed from both.
$$ \models_{KT} (B \cup A \prec C \cup A) \wedge (B \not \perp C) \, \to \, B \prec C$$
\end{Proposition}
\begin{proof}
Suppose on the contrary that in an arbitrary world $s$ of a $KT$ model, we have $s \models \neg (B \prec C)$, then we know from Observation \ref{constructs} that either $s \models C \preceq B$ or $s \models B \perp C$ holds. In~the first case, we use Proposition \ref{KT-version} to obtain $s \models C \cup A \preceq B \cup A$ which implies $s \models \neg (B \cup A \prec C \cup A)$, while the second case establishes the negation of the second conjunct. Hence, we obtain \\ $s \models \neg((B \cup A \prec C \cup A) \wedge (B \not \perp C))$.
\end{proof}

If an epistemic advantage is established by $B$ when it is fused with $A$, over~the group composed of $C$ and $A$, then it will hold up even if $A$ steps out of the coalition with $C$, i.e.,
\begin{Proposition}\label{prop9} In $KT$-models, if~$B \cup A$ knows more than $C \cup A$, then $B \cup A$ also knows more than $C$ alone:
$$\models_{KT} B \cup A \prec C \cup A \, \to \,  B \cup A \prec C$$
\end{Proposition}
\begin{proof}
As $C \subseteq  C \cup A$, the~Inclusion axiom indicates that $\models_{KT} C\cup A \preceq C$. If~$ \models_{KT} B \cup A \prec C \cup A \preceq C$, then by the transitivity axiom we obtain $\models_{KT} B \cup A \prec C \cup A \, \to \,  B \cup A \prec C$.
\end{proof}

In epistemic $KT$ models, a~group $B$ knows at least as much as $C$ and knows at least as much as $E$ if and only if it knows at least as much as $C \cup E$. This indicates that while $B$ may have an epistemic advantage over each one of the groups $C$ and $E$, but~if $C$ and $E$ join epistemic forces, they may well be epistemically as strong as $B$:
\begin{Proposition}\label{pww} In $KT$-models, if~$B$ (collectively) knows at least as much as group $C$ and as group $E$, then $B$ (collectively) knows at least as much as the union of $C$ and $E$:
$$\models_{KT} (B \preceq C \wedge B \preceq E) \, \leftrightarrow \,  (B \preceq C \cup E)$$
\end{Proposition}
\begin{proof}
The direction from left to right follows from the Additivity axiom, and the direction from right to left follows from Proposition \ref{prop9}.
\end{proof}

Interestingly, the~strict version of the last proposition does not hold, i.e.,
\begin{Observation} The analogue of Proposition \ref{pww} does not hold for strict epistemic comparisons, not even in $S5$ models:
$$\not \models_{S5} (B \prec C \wedge B \prec E) \, \leftrightarrow \,  (B \prec C \cup E)$$
\end{Observation}
\begin{proof}
Consider the model $M$ in Fig.~\ref{Fig5.}, with $u \models_{M} ({a} \prec {b}) \wedge ({a} \prec {c})$ while ${u \models_{M} {a} \equiv ({\{b\} \cup \{c\})}}$, hence $u \not \models_{M} {a} \prec (\{b\} \cup \{c\})$. This shows that $(B \prec C \wedge {B \prec E})  \rightarrow   ({B \prec C \cup E})$ fails, i.e.,~the strict version of the Additivity axiom fails in $S5$.
\end{proof}

\smallskip\par\noindent
{\bf Epistemic free-riders.} Given a state of comparative knowledge of a large group $C$ consisting of different sub-groups, we can analyze the relative contribution of any of its subgroups to the knowledge of the overall larger group $C$. For~instance, if~$C$ is split into two disjoint subgroups, $A$ and $B$ (i.e., $C = A \cup B$ with $A \cap B = \emptyset$), then $A \prec B$ means that, in principle, $B$ makes no additional contribution over $A$ to the bigger group $C$ and can be seen as an ``epistemic free-rider''. Here, $B$ does not contribute anything new that is not already ``known'' by $A$. In~these cases when only the knowledge of the group is at stake, the~group $C$ will not gain any added epistemic value by consulting $B$ as long as $A$'s knowledge is available to $C$. In~contrast, when the groups $A$ and $B$ are epistemically equivalent, i.e.,~$A \equiv B$, one of the two groups $A$ or $B$ already covers all the knowledge of the other group, so either one of them can be an epistemic free-rider, but not both, as at least one of them will need to keep its position within $C$.

\subsection{``Known Superiority''~Revisited}

Should the epistemic comparative modality satisfy the ``Known Superiority'' axiom? One can easily object that assuming that the more-expert group always knows their epistemic position is rather unrealistic. The~above results indicate that a strong account of ``Known Superiority'' holds in any epistemic model in which knowledge is assumed to be factive and both positive and negative introspective. However, a~weakening of ``Known Superiority'' (while keeping the other axioms in place) can be achieved by using a softer notion of knowledge, in~which we give up the requirement of negative introspection. So these agents may not know what they do not know. This brings us to the analysis of a weaker version of ``Known Superiority'' based on the knowledge modeled in epistemic $S4$ models:

\smallskip\par\noindent
{\bf{Analysis of Example 2.}} Consider again the $S4$ model of Figure~\ref{Fig2.}. In~world $s$, agent $b$ does not know that he knows more than $a$, indeed he considers all three scenarios possible, namely that he knows more than $a$ but also that $a$ knows more than him or that they are epistemically equivalent.

As such Figure~\ref{Fig2.} provides a proof for the following observation:
\begin{Observation} In $S4$-models, if~$B$ (collectively) knows at least as much as group $C$ then it is not necessarily the case that $B$ knows this fact: 
$$\not \models_{S4} (B \preceq C) \to K_B (B \preceq C)$$
\end{Observation}

However, whenever the less-expert group knows that others know at least as much in epistemic $S4$ and $KT$ models, then so does the more-expert group (as they know everything the less-expert group knows), i.e.,
\begin{Proposition}\label{p16} In $KT$-models, if~group $C$ (collectively) knows that another group $B$ (collectively) knows at least as much as them, then group $B$ also knows it: 
$$\models_{KT} {K_C (B \preceq C) \to K_B (B \preceq C)}$$
\end{Proposition}
\begin{proof}
{{By the $T$-axiom, we have $\models_{KT} K_C (B \preceq C) \to B \preceq C$, and by the ``Knowledge Transfer'' axiom, we have $\models_{KT} B \preceq C \to (K_C (B \preceq C) \to K_B (B \preceq C))$. By~putting these together and using propositional logic, we obtain $\models_{KT} K_C (B \preceq C) \to K_B (B \preceq C)$.}}
\end{proof}

Note that the proof of Proposition \ref{p16} uses the veracity of knowledge but not the axiom for negative introspection, hence the proposition indeed holds in epistemic $KT$ and $S4$ models.

\section{Conclusions} \label{sec4}

In this paper we presented a discussion on comparative knowledge, providing a logical analysis to reason about what agents know about their own epistemic position compared to that of others. Our analysis went beyond the work in~\cite{LPAR}, providing a study of group epistemic comparisons in the more general context of epistemic models of three types ($KT$, $S4$, and $S5$). In~particular, we have shown that the axiom of ``Known Superiority'' is valid in $S5$ and fails in $S4$ models. In~addition we highlighted a series of interesting validities that still do hold in the different models. These insights will be crucial to obtain a complete axiomatization for the logic for comparative knowledge with respect to $S4$ and $KT$ models, which we have left open for future work. In~addition, we envision a number of other different directions in which this work can be taken further:

\bigskip\par\noindent
(a) In 
 \cite{LPAR}, we studied the dynamic epistemic logic with events that lead to the realization of epistemic comparative statements in $S5$-type models. In~particular, we studied a variety of so-called semi-public ``reading'' events in which a given agent or group of agents gains access to the knowledge base or database of another agent or group of agents. These events are called ``semi-public'' when we assume it to be common knowledge who can ``read'' whose database, but~not their content. A~more general class of event-types was studied in~\cite{AIML24}, extending the full power of standard dynamic epistemic logics (see~\cite{BMS,Baltag and Renne:2016,LDII,DHK}), with~features to model agent's data-exchanges, public and private communication. The~study of such dynamic logics for the realization of comparative knowledge with underlying $S4$ or $KT$ models for knowledge, is also left open for future~work.

\bigskip\par\noindent
(b) The study of $S4$-style comparative knowledge was in this paper restricted to relational models, but~it can alternatively be continued in a topological semantics for knowledge, e.g.,~by building on the work in~\cite{TopoPaper}. This would allow for a comparative relation between topologies, that can be used to express a notion of potential knowledge and similarly a potential epistemic comparison~statement.

\bigskip\par\noindent
(c) The semantic clause for epistemic comparison statements in this paper made essential use of the notion of distributed knowledge. 
In ongoing work we are extending the investigation and defining epistemic comparison statements with different underlying epistemic base concepts such as, e.g.,~common knowledge or mutual knowledge. This would for instance replace $R_B$ with $R^B$ (and similar for $C$) in the semantic clause for epistemic comparison statements. Similarly, we are considering a setting that makes the different hierarchical social structures between agents and groups of agents explicit, allowing for a more fine-grained study that can connect to other logical studies of epistemic--social phenomena (e.g., \cite{cascades,diffusion,hom}).

\bigskip\par\noindent
(d) The setting of comparative knowledge in this paper assumes that one group knows at least as much ``about everything'' that another group knows. It can be argued that this is hard to achieve and that it is more realistic to consider agents or groups of agents who ``know more about a specific topic'' when compared to others. In~ongoing work we are exploring ``topic-relative epistemic comparisons''.

\bigskip\par\noindent
(e) Another direction for future work extends the setting with a doxastic notion to express the agent's ``beliefs''. The~combination of different epistemic and doxastic attitudes as well as their dynamics will require a richer type of relational models, e.g., such as those given by the epistemic plausibility models in~\cite{qualitative}. In~such a framework it will be interesting to express the cognitive bias known as ``illusory superiority'', i.e.,~having the wrong belief that your group is the more-expert group. Similarly, we can model scenarios in which groups believe to know that they are the expert group when in fact they are not, as~well as scenarios in which groups make others believe that they are the less-expert group when in fact they are the most-expert~group.

\bigskip\par\noindent
(f) The notions of group knowledge discussed in this paper, such as distributed and common knowledge, can be extended to include the recently introduced concept of common distributed knowledge \cite{LPAR,Suzanne}.\footnote{The concept was first defined (but not axiomatized) in the master's thesis \cite{Suzanne} supervised by A. Baltag, and was independently discussed in \cite{Seligman1,Seligman2}. In \cite{LPAR}, a sound and complete axiomatization for the logic of common distributed knowledge with epistemic comparison is provided and it is shown to be decidable. In the context of dynamic epistemic logic (DEL), the concept is further studied in \cite{LPAR,AIML24,TARK} and continues the earlier line of work on converting distributed knowledge into common knowledge (see \cite{Lonely} and \cite{BS5}).} Common distributed knowledge of a proposition $\varphi$, captures the common knowledge of the distributed knowledge of $\varphi$. More specifically, we use the following definition in~\cite{LPAR}:

\smallskip\par\noindent
{\bf Common distributed knowledge}. Given a family ${\mathcal B}=\{B_1, \ldots, B_n\} \subseteq {\mathcal P}({\cal{A}})$ of groups of agents, we say that $\varphi$ is \emph{common distributed knowledge} among (the groups in) the supergroup ${\mathcal B}$, if~we have that: each group $B \in {\mathcal B}$ has distributed knowledge that $\varphi$; each group $B \in {\mathcal B}$ has distributed knowledge that each other group $B' \in {\mathcal B}$ has distributed knowledge that $\varphi$; etc. (for all iterations). Formally:
$$s\models K^{\mathcal B}\varphi \,\, \mbox{ iff }\,\, s\models K_{B_1} K_{B^2}\ldots K_{B^n} \varphi \mbox{ for all sequences (of any length $n\geq 0$) } B_1, \ldots, B_n\in {\mathcal B}.$$

We take $K^{\mathcal B}$ to be the Kripke modality for the relation $R^{\mathcal B}$, given by
$$R^{\mathcal B}\,\, :=\,\, (\bigcup_{B\in {\mathcal B}}R_B)^*$$
(where, as before, $R^*$ is the reflexive--transitive closure of $R$). Note that our earlier notation for the common knowledge possibility relation $R^B$ is now generalized from a single group of agents $B$ to a supergroup (i.e., a set consisting of different groups of agents) $\mathcal{B}$, yielding $R^{\mathcal B}$.

As can be seen from this definition, common distributed knowledge generalizes the notions of common knowledge and distributed knowledge.  Indeed, for~a given family $\{B_1, \ldots, B_n\}$ of groups of agents, we see that common distributed knowledge (here abbreviated as $K^{\{B_1, \ldots, B_n\}}$) is weaker than common knowledge but stronger than distributed knowledge, as~it has features of both. The~following proposition explains this further:

\begin{Proposition}
The following schematic diagram illustrates the strength relationships between the various group attitudes mentioned in this paper (where we denote by $\Longrightarrow$ the entailment relation, i.e.,~valid implication, between~two statements):
$$K^{B_1 \cup \ldots \cup B_n} \varphi \Longrightarrow K^{\{B_1, \ldots, B_n\}} \varphi \Longrightarrow (K_{B_1} \varphi \wedge K_{B_2} \varphi \wedge \ldots \wedge K_{B_n} \varphi) \Longrightarrow K_{B_1 \cup \ldots \cup B_n}\varphi $$
\end{Proposition}

Given a logical framework with common distributed knowledge, allows for the different notions of distributed knowledge, common knowledge, and individual knowledge to be syntactically defined as special cases. While common distributed knowledge is used in~\cite{LPAR} with the specific aim of obtaining reduction laws in a logical system, this epistemic notion is of independent interest and leads to new insights in a study that combines it with the features of comparative knowledge that we have presented in this paper. In~particular, it leads to the following generalization to groups of agents for Propositions \ref{p17}--\ref{proposition9c}:

\begin{Proposition} In $S4$ models, if~a group $C$ (collectively) knows that another group $B$ (collectively) knows at least as much as them, then~this is also common distributed knowledge in the supergroup consisting of $B$ and $C$. In~$S5$ models, if~two groups $B$ and $C$ are (not, respectively) epistemically equivalent, then this fact is common distributed knowledge in the supergroup consisting of $B$ and $C$:
$$\models_{S4} K_C(B \preceq C) \to K^{\{B,C\}}(B \preceq C)$$
$$\models_{S5} B \equiv C \, \to \, K^{\{B,C\}}(B \equiv C)$$
$$ \models_{S5} B \not \equiv C \, \to \, K^{\{B,C\}}(B \not \equiv C)$$
\end{Proposition}
\par\noindent
The proofs of these generalized statements for common distributed knowledge in the $S5$ case make use of the valid axioms for common distributed knowledge provided in~\cite{LPAR}; while for $S4$ we drop the negative introspection axiom for common distributed knowledge. This opens the door to further studies on what a supergroup commonly can know about its own relative epistemic power.


\bigskip\par\noindent
{\bf Acknowledgments} We thank the organizers and participants of the 3rd International Workshop on Logic and Philosophy (held at Tsinghua University in 2024) that gave rise to the special issue on Collective Agency and Intentionality, where this work on epistemic comparison is presented.


\end{document}